\documentclass[journal]{IEEEtran}

\usepackage[T1]{fontenc}
\usepackage[utf8]{inputenc}   

\usepackage{amsmath, amssymb, amsthm}

\usepackage{braket}           

\usepackage{graphicx}
\usepackage{tikz}
\usetikzlibrary{shapes.geometric, arrows.meta, positioning}

\usepackage{array}
\usepackage{booktabs}

\newtheorem{proposition}{Proposition}

\ifCLASSINFOpdf
\else
\fi

\begin{document}

\title{An Effective Pauli-Channel Model for Passive-User Loop-Back QKD}

\author{Luis~Adrián~Lizama-Pérez%
\thanks{Luis Adrián Lizama-Pérez is with the Departamento de Sistemas de Información y Comunicaciones,
Universidad Autónoma Metropolitana, Unidad Lerma, Lerma, Estado de México, Mexico
(e-mail: l.lizama@correo.ler.uam.mx).}}%

\markboth{IEEE Communications Letters}%
{Lizama-Pérez: An Effective Pauli-Channel Model for Passive-User Loop-Back QKD}

\maketitle

\begin{abstract}
This letter develops an effective channel model for distributed
passive-user Loop-Back quantum key distribution. 
In the intended key-establishment setting, the two passive users \(B_1\) and \(B_2\) are the legitimate lightweight parties that establish a shared secret key by using Alice as an active quantum preparation-and-measurement infrastructure.
A single active
station prepares and measures BB84 states, while two remote users
apply only local polarization rotations. We show that the passive-user
pair can be externally encapsulated as an effective Loop-Back node
whose observable action is a balanced mixture of conjugate-basis
dephasings, equivalently represented as an anisotropic Pauli channel
with identity, \(X\), and \(Z\) components and no \(Y\)-component.
This structure differs from isotropic depolarization and recovers the
ideal conclusive-event probability \(P_{\mathrm{conc}}=1/4\). The
model also clarifies why non-orthogonal intermediate states are
necessary for passive-user security. This channel-level description characterizes Alice's observable statistics and provides a compact basis for subsequent analysis of passive-user Loop-Back QKD under realistic optical and adversarial conditions.
\end{abstract}


\IEEEpeerreviewmaketitle

\section{Introduction}

\IEEEPARstart{Q}{uantum} key distribution (QKD) provides a physical-layer mechanism for establishing secret correlations whose security does not rely on computational assumptions. Since BB84~\cite{bennett2014quantum}, most practical QKD architectures have followed a prepare-and-measure paradigm in which the communicating users require active quantum functionality, including state preparation, basis selection, synchronization, and, in many cases, single-photon detection. Although this model is well suited to laboratory demonstrations and backbone links, it remains demanding for lightweight, mobile, or densely deployed user nodes~\cite{scarani2009security}.

A natural strategy to reduce this burden is to concentrate the active quantum functionality at a single station while restricting the remote side to simpler optical operations. Related asymmetric and two-way architectures include plug-and-play systems~\cite{muller1997plug}, semi-quantum protocols~\cite{boyer2007quantum,boyer2009semiquantum}, and deterministic two-way schemes~\cite{lucamarini2005secure}. Loop-Back QKD follows this principle: an active node prepares a quantum state, a remote node modifies or reflects the optical signal, and the state is returned for measurement~\cite{lizama2025loop,lizama2025three}. The relevant information is then associated with the accumulated round-trip transformation rather than with a direct receiver-side measurement.

This letter considers a distributed passive-user version of this setting. A single active station, Alice, prepares and measures BB84 polarization states, whereas two remote users, \(B_1\) and \(B_2\), do not prepare or detect quantum states. Instead, each passive user applies a local polarization rotation to the same propagating quantum carrier. The operation observed by Alice is therefore the composed transformation
\begin{equation}
    R_{\mathrm{eff}} = R_2R_1 .
\end{equation}
For rotations selected from \(\{R(+\pi/8),R(-\pi/8)\}\), this effective operation either cancels to the identity or produces a rotation toward the conjugate basis. Hence, conclusive events arise from the observable action of the composed passive subsystem, while the individual operations remain internal to the passive users.

The main contribution of this letter is to show that this distributed passive-user mechanism admits a compact effective-channel representation. From Alice's viewpoint, the pair \(B^*=(B_1,B_2)\) can be encapsulated as a single effective Loop-Back node whose action is a balanced mixture of dephasings in conjugate bases,
\begin{equation}
    \mathcal{E}_{\mathrm{LB}}(\rho)
    =
    \frac{1}{2}D_Z(\rho)
    +
    \frac{1}{2}D_X(\rho).
\end{equation}
We show that this map is equivalently an anisotropic Pauli channel with identity, \(X\), and \(Z\) components, and no \(Y\)-component. The representation recovers the ideal conclusive-event probability and shows that the induced disturbance is not isotropic depolarization, but an intrinsic channel generated by the distributed transformation structure.

This channel-level description characterizes Alice's measurement statistics, but it is not a complete composable-security proof. Security of the passive-user implementation still depends on the absence of perfectly distinguishable intermediate states and on the resulting information--disturbance trade-off~\cite{wootters1982single,fuchs1996quantum}.

\section{Distributed Passive-User Loop-Back Model}

We consider a minimal distributed Loop-Back configuration consisting of one active station, Alice, and two optically passive users, \(B_1\) and \(B_2\). Alice is the only node equipped with quantum state preparation and measurement capabilities. The passive users do not generate, measure, or store quantum states; they only apply local polarization rotations to the same propagating optical carrier.

A single protocol round follows the closed optical path
\begin{equation}
    A \longrightarrow B_1 \longrightarrow B_2 \longrightarrow A .
\end{equation}
Alice prepares a BB84 polarization state
\begin{equation}
    \ket{\psi_{b,k}} \in 
    \{\ket{0_Z},\ket{1_Z},\ket{0_X},\ket{1_X}\},
\end{equation}
where \(b\in\{Z,X\}\) denotes the preparation basis and \(k\in\{0,1\}\) denotes the encoded state within that basis. The state is sent through the two passive users and then returned to Alice, who measures it in the same basis \(b\) used for preparation.

Each passive user \(B_i\), with \(i\in\{1,2\}\), independently selects one of two polarization rotations, represented as qubit rotations in the corresponding Bloch-sphere plane,
\begin{equation}
    R_i = R(\theta_i),
    \qquad
    \theta_i \in \left\{+\frac{\pi}{8},-\frac{\pi}{8}\right\}.
\end{equation}
The state returning to Alice is therefore
\begin{equation}
    \ket{\psi_{\mathrm{out}}}
    =
    R_2 R_1 \ket{\psi_{b,k}} .
\end{equation}
Since the rotations act in the same polarization plane, their composition is additive:
\begin{equation}
    R_2R_1 = R(\theta_1+\theta_2).
\end{equation}
Thus, the effective rotation angle satisfies
\begin{equation}
    \theta_{\mathrm{eff}}
    =
    \theta_1+\theta_2
    \in
    \left\{
        -\frac{\pi}{4},\,0,\,+\frac{\pi}{4}
    \right\}.
\end{equation}

There are two operational cases. If the users select opposite rotations, then
\begin{equation}
    \theta_1=-\theta_2,
    \qquad
    R_2R_1=R(0)=I,
\end{equation}
and the state ideally returns unchanged to Alice. These rounds are inconclusive because Alice obtains the state originally prepared and cannot identify a nontrivial composed transformation.

If the users select the same rotation, then
\begin{equation}
    \theta_1=\theta_2,
    \qquad
    R_2R_1=R\left(\pm\frac{\pi}{4}\right).
\end{equation}
In this case, the input state is rotated toward the conjugate basis. For example, for \(\ket{\psi_{b,k}}=\ket{0_Z}\),
\begin{equation}
    R\left(+\frac{\pi}{4}\right)\ket{0_Z}
    =
    \ket{0_X},
    \qquad
    R\left(-\frac{\pi}{4}\right)\ket{0_Z}
    =
    \ket{1_X}.
\end{equation}
A subsequent measurement by Alice in the original \(Z\) basis yields either \(\ket{0_Z}\) or \(\ket{1_Z}\) with equal probability. The observation of the state orthogonal to the one prepared is therefore a conclusive event, because it certifies that a nontrivial effective rotation was applied.

More generally, if Alice prepares \(\ket{\psi_{b,k}}\) and measures in the same basis \(b\), the conclusive-event rule is
\begin{equation}
    \mathrm{conclusive}
    \iff
    \Pi_m =
    \ket{\psi_{b,k}^{\perp}}\!\bra{\psi_{b,k}^{\perp}},
\end{equation}
where \(\Pi_m\) is the projector associated with Alice's measurement outcome \(m\), and \(\ket{\psi_{b,k}^{\perp}}\) denotes the state orthogonal to \(\ket{\psi_{b,k}}\) in basis \(b\). Under ideal balanced choices,
\begin{equation}
    \Pr(\theta_1=\theta_2)=\frac{1}{2},
\end{equation}
and, conditioned on this event,
\begin{equation}
    \Pr(m=\psi_{b,k}^{\perp}\mid \theta_1=\theta_2)=\frac{1}{2}.
\end{equation}
Therefore, the ideal conclusive-event probability is
\begin{equation}
    \Pr(\mathrm{conclusive})
    =
    \frac{1}{2}\cdot\frac{1}{2}
    =
    \frac{1}{4}.
    \label{eq:conclusive_probability}
\end{equation}

This probability is not introduced here as a final key rate. Rather, it characterizes the observable Loop-Back statistics produced by the distributed passive subsystem. In the following section, we show that these same statistics can be obtained by externally encapsulating \(B_1\) and \(B_2\) as a single effective Loop-Back node.

After the quantum transmission, Alice publicly announces only the
indices of the conclusive rounds, over an authenticated classical
channel. Ambiguous or inconclusive rounds are discarded. Alice does
not reveal her prepared state or measurement outcome for the retained
rounds. Hence, in the ideal case, the announcement only performs
sifting. For each retained round, the condition of conclusiveness
implies \(\theta_1=\theta_2\), so that each passive user can keep the
bit associated with its own local rotation as a raw shared-key bit.

Alice's conclusive announcement reveals only that a nontrivial
effective rotation occurred, equivalently \(\theta_1=\theta_2\), but
it does not reveal whether the common local choice was \(+\pi/8\) or
\(-\pi/8\). Hence, when the raw bit is assigned to the sign of the
local rotation, \(B_1\) and \(B_2\) share the same raw-key bit, whereas
Alice does not have access to its value.

\section{Effective Loop-Back Channel Observed by Alice}

The distributed passive-user subsystem can be externally represented as a single effective Loop-Back node. From Alice's viewpoint, the internal decomposition into \(B_1\) and \(B_2\) is not directly accessible; only the global transformation applied to the returning quantum state is observed. We therefore define the effective passive subsystem
\begin{equation}
    B^*=(B_1,B_2),
\end{equation}
whose action is determined by
\begin{equation}
    R_{\mathrm{eff}}
    =
    R_2R_1
    =
    R(\theta_1+\theta_2).
\end{equation}
As shown in Section~II, the effective angle satisfies
\begin{equation}
    \theta_{\mathrm{eff}}
    \in
    \left\{
        -\frac{\pi}{4},\,0,\,+\frac{\pi}{4}
    \right\}.
\end{equation}
The case \(\theta_{\mathrm{eff}}=0\) leaves Alice's prepared state unchanged, whereas the cases \(\theta_{\mathrm{eff}}=\pm\pi/4\) map it toward the conjugate basis. Hence, the effective subsystem behaves as a node that either preserves Alice's preparation basis or induces a conjugate-basis action.

To make this relation explicit, suppose that Alice prepares \(\ket{0_Z}\). In the identity case,
\begin{equation}
    R(0)\ket{0_Z}=\ket{0_Z}.
\end{equation}
In the nontrivial cases,
\begin{equation}
    R\left(+\frac{\pi}{4}\right)\ket{0_Z}
    =
    \ket{0_X},
    \qquad
    R\left(-\frac{\pi}{4}\right)\ket{0_Z}
    =
    \ket{1_X}.
\end{equation}
Averaging over the two equally likely signs gives
\begin{equation}
    \frac{1}{2}\ket{0_X}\!\bra{0_X}
    +
    \frac{1}{2}\ket{1_X}\!\bra{1_X}
    =
    \frac{I}{2}.
\end{equation}
Thus, conditioned on a nontrivial effective rotation, the state returned to Alice is statistically indistinguishable from a state completely dephased in the conjugate basis.

The same reasoning applies to any BB84 input state
\begin{equation}
    \rho_{b,k}
    =
    \ket{\psi_{b,k}}\!\bra{\psi_{b,k}},
    \qquad
    b\in\{Z,X\},
    \quad
    k\in\{0,1\}.
\end{equation}
Let \(D_b\) denote complete dephasing in Alice's preparation basis \(b\), and let \(D_{\bar b}\) denote complete dephasing in the conjugate basis \(\bar b\). Thus, if \(b=Z\), then \(D_b=D_Z\) and \(D_{\bar b}=D_X\), whereas if \(b=X\), then \(D_b=D_X\) and \(D_{\bar b}=D_Z\). Since \(\rho_{b,k}\) is an eigenstate of basis \(b\),
\begin{equation}
    D_b(\rho_{b,k})
    =
    \rho_{b,k}.
\end{equation}
By contrast, complete dephasing in the conjugate basis gives
\begin{equation}
    D_{\bar b}(\rho_{b,k})
    =
    \frac{1}{2}\rho_{b,k}
    +
    \frac{1}{2}\rho_{b,k}^{\perp},
\end{equation}
where
\begin{equation}
    \rho_{b,k}^{\perp}
    =
    \ket{\psi_{b,k}^{\perp}}\!\bra{\psi_{b,k}^{\perp}} .
\end{equation}

Under balanced passive-user choices, the effective Loop-Back action observed by Alice is therefore
\begin{equation}
    \mathcal{E}_{\mathrm{LB}}(\rho_{b,k})
    =
    \frac{1}{2}D_b(\rho_{b,k})
    +
    \frac{1}{2}D_{\bar b}(\rho_{b,k}).
\end{equation}
Substituting the two dephasing contributions yields
\begin{align}
    \mathcal{E}_{\mathrm{LB}}(\rho_{b,k})
    &=
    \frac{1}{2}\rho_{b,k}
    +
    \frac{1}{2}
    \left(
        \frac{1}{2}\rho_{b,k}
        +
        \frac{1}{2}\rho_{b,k}^{\perp}
    \right)  \nonumber\\
    &=
    \frac{3}{4}\rho_{b,k}
    +
    \frac{1}{4}\rho_{b,k}^{\perp}.
    \label{eq:effective_lb_on_bb84}
\end{align}
Consequently, when Alice measures in her original preparation basis, the probability that the outcome is associated with the orthogonal projector is
\begin{equation}
    \Pr\!\left(
    \Pi_m=\rho_{b,k}^{\perp}
    \right)
    =
    \frac{1}{4},
\end{equation}
which coincides with the conclusive-event probability obtained directly from the distributed transformation rule in \eqref{eq:conclusive_probability}.

Equation~\eqref{eq:effective_lb_on_bb84} shows that the distributed passive-user subsystem reproduces, at the level of Alice's observable statistics, the conclusive-event structure of an effective Loop-Back node. The distinction is internal: in the distributed implementation, the effective operation is not selected by a single remote party, but emerges from the private composition of two local transformations. Alice can identify whether a conclusive event occurred, but the internal decomposition of \(R_{\mathrm{eff}}\) remains inaccessible from her measurement record alone.

\section{Anisotropic Pauli-Channel Representation}

The effective Loop-Back channel derived in the previous section was obtained from Alice's prepare-and-measure statistics on BB84 states. We now introduce a Pauli-channel representative for the externally encapsulated Loop-Back box \(B^{\ast}=(B_1,B_2)\). This representative reproduces Alice's BB84 preparation-and-measurement statistics and provides a compact external model of the passive-user subsystem, without resolving the coherent internal sequence of transformations between \(B_1\) and \(B_2\). This representation makes explicit that the induced map is not an isotropic depolarizing channel, but a structured anisotropic Pauli channel associated with the two conjugate bases involved in the Loop-Back interaction.

Let
\begin{equation}
    \rho
    =
    \frac{1}{2}
    \left(
        I + r_x X + r_y Y + r_z Z
    \right)
\end{equation}
be an arbitrary qubit state, where \(X\), \(Y\), and \(Z\) are the Pauli matrices and
\begin{equation}
    \mathbf{r}=(r_x,r_y,r_z)
\end{equation}
is its Bloch vector. Complete dephasing in the \(Z\) basis removes the components transverse to the \(Z\) axis:
\begin{equation}
    D_Z(\rho)
    =
    \frac{1}{2}
    \left(
        I + r_z Z
    \right).
\end{equation}
Similarly, complete dephasing in the \(X\) basis removes the components transverse to the \(X\) axis:
\begin{equation}
    D_X(\rho)
    =
    \frac{1}{2}
    \left(
        I + r_x X
    \right).
\end{equation}

The effective Loop-Back operation observed by Alice is the balanced mixture
\begin{equation}
    \mathcal{E}_{\mathrm{LB}}(\rho)
    =
    \frac{1}{2}D_Z(\rho)
    +
    \frac{1}{2}D_X(\rho).
    \label{eq:lb_dephasing_mixture}
\end{equation}
Substituting the expressions of \(D_Z\) and \(D_X\) into \eqref{eq:lb_dephasing_mixture}, we obtain
\begin{align}
    \mathcal{E}_{\mathrm{LB}}(\rho)
    &=
    \frac{1}{2}
    \left[
        \frac{1}{2}
        \left(
            I+r_z Z
        \right)
    \right]
    +
    \frac{1}{2}
    \left[
        \frac{1}{2}
        \left(
            I+r_x X
        \right)
    \right] \nonumber\\
    &=
    \frac{1}{2}I
    +
    \frac{1}{4}r_x X
    +
    \frac{1}{4}r_z Z .
    \label{eq:lb_bloch_action}
\end{align}
Thus, in the Bloch representation, the effective Loop-Back channel acts as
\begin{equation}
    (r_x,r_y,r_z)
    \longmapsto
    \left(
        \frac{r_x}{2},\,0,\,\frac{r_z}{2}
    \right).
    \label{eq:lb_bloch_map}
\end{equation}
The \(Y\)-component is completely suppressed, whereas the \(X\)- and \(Z\)-components are contracted by a factor \(1/2\). This already shows that the channel is anisotropic.

\begin{proposition}
The balanced effective Loop-Back channel
\begin{equation}
    \mathcal{E}_{\mathrm{LB}}(\rho)
    =
    \frac{1}{2}D_Z(\rho)
    +
    \frac{1}{2}D_X(\rho)
\end{equation}
is equivalent to the anisotropic Pauli channel
\begin{equation}
    \mathcal{E}_{\mathrm{LB}}(\rho)
    =
    \frac{1}{2}\rho
    +
    \frac{1}{4}X\rho X
    +
    \frac{1}{4}Z\rho Z.
    \label{eq:lb_pauli_channel}
\end{equation}
\end{proposition}

\begin{proof}
Consider the general Pauli channel
\begin{equation}
    \mathcal{P}(\rho)
    =
    p_I\rho
    +
    p_X X\rho X
    +
    p_Y Y\rho Y
    +
    p_Z Z\rho Z .
\end{equation}
Set
\begin{equation}
    p_I=\frac{1}{2},
    \qquad
    p_X=\frac{1}{4},
    \qquad
    p_Y=0,
    \qquad
    p_Z=\frac{1}{4}.
\end{equation}
These coefficients satisfy \(p_I+p_X+p_Y+p_Z=1\), and the map is completely positive and trace preserving because it is a convex mixture of unitary Pauli conjugations. Then
\begin{equation}
    \mathcal{P}(\rho)
    =
    \frac{1}{2}\rho
    +
    \frac{1}{4}X\rho X
    +
    \frac{1}{4}Z\rho Z .
    \label{eq:pauli_channel_form}
\end{equation}
Using
\begin{align}
    X\rho X
    &=
    \frac{1}{2}
    \left(
        I + r_x X - r_y Y - r_z Z
    \right),\\
    Z\rho Z
    &=
    \frac{1}{2}
    \left(
        I - r_x X - r_y Y + r_z Z
    \right),
\end{align}
and substituting into \eqref{eq:pauli_channel_form}, we obtain
\begin{align}
    \mathcal{P}(\rho)
    &=
    \frac{1}{2}
    \left[
        \frac{1}{2}
        \left(
            I+r_xX+r_yY+r_zZ
        \right)
    \right] \nonumber\\
    &\quad+
    \frac{1}{4}
    \left[
        \frac{1}{2}
        \left(
            I+r_xX-r_yY-r_zZ
        \right)
    \right] \nonumber\\
    &\quad+
    \frac{1}{4}
    \left[
        \frac{1}{2}
        \left(
            I-r_xX-r_yY+r_zZ
        \right)
    \right] \nonumber\\
    &=
    \frac{1}{2}I
    +
    \frac{1}{4}r_xX
    +
    \frac{1}{4}r_zZ.
\end{align}
This coincides with the Bloch action in \eqref{eq:lb_bloch_action}. Therefore,
\begin{equation}
    \mathcal{P}(\rho)
    =
    \mathcal{E}_{\mathrm{LB}}(\rho),
\end{equation}
which proves the claim.
\end{proof}

Equation~\eqref{eq:lb_pauli_channel} provides a compact channel-level description of the passive-user Loop-Back operation. The identity component represents the part of the process compatible with Alice's preparation basis, whereas the \(X\) and \(Z\) components encode basis-dependent flips associated with the conjugate-basis branch. The absence of a \(Y\)-component reflects the fact that, in the ideal model, the effective process does not combine simultaneous bit- and phase-flip behavior.

This distinction is operationally relevant. In an isotropic depolarizing channel, the three Pauli components appear symmetrically. Here, by contrast, the channel preserves the special role of the \(X\) and \(Z\) bases used by the Loop-Back interaction. The effective noise is therefore not an external generic disturbance, but the intrinsic channel induced by the distributed passive transformation structure.

\section{Secret-Key-Rate Interpretation and Intermediate-State Security}

The Pauli-channel representation in \eqref{eq:lb_pauli_channel} explains why the ideal conclusive-event probability is
\begin{equation}
    P_{\mathrm{conc}}=\frac{1}{4}
\end{equation}
under balanced passive-user choices and ideal optical operation. This probability is not, by itself, a secret-key rate; it is the sifting factor associated with the effective Loop-Back inference rule.

A conservative rate interpretation can be written in the generic form
\begin{equation}
    R_{\mathrm{key}}
    =
    P_{\mathrm{conc}}
    \left[
        1
        -
        f_{\mathrm{EC}} h(e_{\mathrm{obs}})
        -
        \Delta_{\mathrm{sec}}(e_{\mathrm{obs}})
    \right],
    \label{eq:key_rate_interpretation}
\end{equation}
where \(h(\cdot)\) is the binary entropy function, \(e_{\mathrm{obs}}\)
is the error rate inferred from the statistics of accepted conclusive
events observed by Alice, \(f_{\mathrm{EC}}\geq 1\) accounts for
error-correction inefficiency, and \(\Delta_{\mathrm{sec}}\) denotes
the privacy-amplification penalty associated with adversarial
information leakage. In the ideal noiseless case, \(e_{\mathrm{obs}}=0\),
and \(P_{\mathrm{conc}}=1/4\) gives the maximum raw conclusive-event
fraction per transmitted signal. Under nonideal conditions, losses,
detector imperfections, and adversarial disturbance reduce the rate
through the observed-error and privacy terms.

Equation~\eqref{eq:key_rate_interpretation} should be interpreted as a
conservative rate estimate associated with the effective Loop-Back
channel observed by Alice. This channel characterizes the input--output
behavior after the composed transformation
\(R_{\mathrm{eff}}=R_2R_1\) has been applied, while the internal
sequence of passive transformations remains physically distributed
between \(B_1\) and \(B_2\). Consequently, the external channel seen by
Alice must be distinguished from the intermediate-state structure that
generates it.

This distinction is essential for passive-user Loop-Back QKD. The protocol relies on the fact that the rotations \(R(\pm\pi/8)\) do not produce perfectly distinguishable intermediate states for the BB84 input ensemble. Hence, an adversary attempting to identify an individual passive-user operation cannot do so through a single disturbance-free measurement compatible with all possible preparations. Any information gain about the local transformation choices is therefore constrained by measurement incompatibility and is expected to induce observable perturbations in Alice's statistics.

To illustrate why the choice \(R(\pm\pi/8)\) is essential, consider the counterexample of a deterministic variant based on larger rotations. If the passive users apply \(R(\pm\pi/2)\), then for an input such as \(\ket{0_Z}\) the intermediate states after the first passive user may become
\begin{equation}
    R\left(+\frac{\pi}{2}\right)\ket{0_Z}
    =
    \ket{0_X},
    \qquad
    R\left(-\frac{\pi}{2}\right)\ket{0_Z}
    =
    \ket{1_X}.
\end{equation}
These states are orthogonal and belong to the same measurement basis, allowing an adversary with access to the intermediate channel to distinguish the local operation without necessarily producing detectable disturbance at that stage. This illustrates why the reduced conclusive-event probability of the \(\pm\pi/8\) construction is not merely an efficiency limitation, but is linked to the use of non-orthogonal intermediate states.

Accordingly, the effective Pauli-channel model should be understood as an operational representation rather than as a replacement for a composable security proof. A complete security analysis would require a quantitative bound relating Eve's accessible information about the passive users' local choices to deviations in Alice's observed statistics, including finite-size effects, losses, detector imperfections, and coherent attacks.

\section{Conclusion}

This letter introduced an effective channel model for distributed
passive-user Loop-Back QKD. We showed that two optically passive users
applying local polarization rotations can be externally encapsulated as
a single effective Loop-Back node. From the active station's viewpoint,
the resulting operation is a balanced mixture of dephasings in
conjugate bases and is equivalently represented as an anisotropic Pauli
channel with identity, \(X\), and \(Z\) components and no
\(Y\)-component.

The representation recovers the ideal conclusive-event probability
\(P_{\mathrm{conc}}=1/4\) and shows that the induced disturbance is not
isotropic depolarization, but an intrinsic channel generated by the
passive composition mechanism. In this setting, Alice acts as an active
quantum infrastructure that identifies and announces the conclusive
rounds, while the raw secret bit remains determined by the private
local choices of \(B_1\) and \(B_2\). The model also separates the
externally observed channel from the intermediate-state structure that
generates it, highlighting the need for non-perfectly distinguishable
local transformations in secure passive-user implementations. This
provides a compact basis for subsequent channel-based QKD analysis
under realistic optical imperfections and adversarial models.
```

\ifCLASSOPTIONcaptionsoff
  \newpage
\fi

\bibliographystyle{IEEEtran}
\bibliography{references}

\end{document}